\documentclass[11pt]{amsart}
\usepackage{amssymb}
\usepackage{tikz,float}
\usepackage{enumerate,url}
\theoremstyle{plain}
\newtheorem{theorem}{Theorem}[section]
\newtheorem{lemma}[theorem]{Lemma}
\newtheorem{proposition}[theorem]{Proposition}
\newtheorem{corollary}[theorem]{Corollary}
\theoremstyle{remark}
\newtheorem{remark}[theorem]{Remark}
\theoremstyle{definition}
\newtheorem{assumption}{Assumption}

\newtheorem{definition}[theorem]{Definition}
\newcommand{\be}{\begin{equation}}
\newcommand{\ee}{\end{equation}}

\newcommand{\ZZ}{\mathbb{Z}}
\newcommand{\RR}{\mathbb{R}}
\newcommand{\CC}{\mathbb{C}}
\newcommand{\NN}{\mathbb{N}}
\newcommand{\EE}{\mathbb{E}}
\newcommand{\DD}{\mathbb{D}}
\newcommand{\cL}{\mathcal{L}}

\newcommand{\cN}{\mathcal{N}}
\newcommand{\cA}{\mathcal{A}}
\newcommand{\Pro}{\ensuremath{P}}
\newcommand{\drm}{\ensuremath{\mathrm{d}}}
\newcommand{\abs}[1]{\left\vert#1\right\vert}
\newcommand{\sprod}[2]{\left\langle#1,#2 \right \rangle}

\newcommand{\PP}{\mathbb{P}}

\newcommand{\euler}{\mathrm{e}}

\DeclareMathOperator{\supp}{\operatorname{supp}}

\DeclareMathOperator{\im}{Im}
\DeclareMathOperator{\Tr}{Tr}
\DeclareMathOperator{\dist}{dist}
\DeclareMathOperator{\diam}{diam}
\DeclareMathOperator{\rank}{rank}

\renewcommand{\epsilon}{\varepsilon}
\renewcommand{\i}{\ensuremath{{\mathrm{i}}}}
\newcommand{\BIGOP}[1]{\mathop{\mathchoice%
{\raise-0.22em\hbox{\huge $#1$}}%
{\raise-0.05em\hbox{\Large $#1$}}{\hbox{\large $#1$}}{#1}}}
\newcommand{\BIGboxplus}{\mathop{\mathchoice%
{\raise-0.35em\hbox{\huge $\boxplus$}}%
{\raise-0.15em\hbox{\Large $\boxplus$}}{\hbox{\large $\boxplus$}}{\boxplus}}}
\newcommand{\bigtimes}{\BIGOP{\times}}
\newcommand{\hm}[1]{\leavevmode{\marginpar{\tiny%
$\hbox to 0mm{\hspace*{-0.5mm}$\leftarrow$\hss}%
\vcenter{\vrule depth 0.1mm height 0.1mm width \the\marginparwidth}%
\hbox to
0mm{\hss$\rightarrow$\hspace*{-0.5mm}}$\\\relax\raggedright #1}}}
\usepackage{microtype}
\begin{document}
%
%
%
%
%
%

\title[Discrete alloy-type models]
{Discrete Schr\"odinger operators with random alloy-type potential}
\author[A. Elgart]{Alexander Elgart}
\address{448 Department of Mathematics, McBryde Hall, Virginia Tech., Blacksburg, VA, 24061, USA}
\urladdr{http://www.math.vt.edu/people/aelgart/}
\author[H. Kr\"uger]{Helge Kr\"uger}
\address{Caltech Department of Mathematics, Mailcode 253-37, 1200 E California Blvd,
Pasadena CA 91125, USA}
\urladdr{http://www.its.caltech.edu/~helge/}
\author[M. Tautenhahn]{Martin Tautenhahn}
\address{Technische Universit\"a{}t Chemnitz \\ Fakult\"a{}t f\"u{}r Mathematik \\ D-09107 Chemnitz \\ Germany}
\urladdr{http://www-user.tu-chemnitz.de/~mtau/}
\author[I. Veseli\'c]{Ivan Veseli\'c}
\address{Technische Universit\"a{}t Chemnitz \\ Fakult\"a{}t f\"u{}r Mathematik \\ D-09107 Chemnitz \\ Germany}
 \urladdr{http://www.tu-chemnitz.de/mathematik/stochastik/}
\begin{abstract}
We review recent results on localization for discrete alloy-type models based on 
the multiscale analysis and the fractional moment method, respectively. 
The discrete alloy-type model is a family of Schr\"odinger operators $H_\omega = - \Delta + V_\omega$ on $\ell^2 (\ZZ^d)$ where $\Delta$ is the discrete Laplacian and $V_\omega$ 
the multiplication by the function $V_\omega (x) = \sum_{k \in \ZZ^d} \omega_k u(x-k)$. 
Here $\omega_k$, $k \in \ZZ^d$,  are i.i.d. random variables and $u \in \ell^1 (\ZZ^d ; \RR)$ is a so-called single-site potential. 
Since $u$ may change sign, certain properties of $H_\omega$ depend in a non-monotone way on the random parameters $\omega_k$.
This requires new methods at certain stages of the localization proof. 
\end{abstract}

%
\subjclass[2000]{82B44, 60H25, 35J10}
 \keywords{multiscale analysis, fractional moment method, Cartan's Theorem, localization, discrete alloy-type model, non-monotone, sign-indefinite, single-site potential}
\maketitle
\tableofcontents
%
%
%
%
%
\section{Introduction}
The paradigmatic model for the study of localization properties of
quantum states  of single electrons  in disordered solids is the Anderson Hamiltonian on the lattice $\ZZ^d$.
It consists of the sum of the finite difference Laplacian and a multiplication operator by a sequence of independent identically distributed random variables.
There are two independent (though related in sprit) frequently used methods
to prove rigorously in appropriate energy and disorder regimes localization statements: the multiscale analysis and the 
fractional moment method. Both of them rely strongly on the independence property of the random variables.
If this property is removed much less is know. Specific cases of random fields with correlations which have been studied so far 
include the Gaussian field (cf.  Section 4 of \cite{DreifusK-91}  and Section 4 of \cite{AizenmanM-93}) and 
potentials whose distribution is given by a completely analytic Gibbs measure
(cf.  Section 5 of \cite{DreifusK-91}  and \cite{DobrushinS-87}).
\par
A particularly problematic aspect of dependence are negative correlations 
between values of the potential at different lattice sites.
This is intuitively clear when having in mind the role played by the local variation of eigenvalues 
in localization proofs.
An example of a random potential which exhibits negative correlations is an Anderson or alloy-type potential with 
single-site potentials of changing sign: If one increases a single coupling constant there are regions in space where the potential
increases and others where it decreases.
Such models have been studied in a number of works devoted to the continuum setting, i.e.{} for operators on $\RR^d$. 
In this paper we want to summarize and discuss a number of results which have been obtained more recently for their lattice counterparts.
\par
The study of Anderson-type models with sign-indefinite single-site potentials can be seen 
as part of the interest attracted recently by several classes of random operators with a non-monotone dependence on the 
random variables. They exhibit a variety  of intriguing features not encountered in the standard alloy-type model:
Already when considering the very basic features of the spectrum as a subset of the real line,
one sees that it makes an essential difference whether the operator depends in a monotone or non-monotone way on the random variables entering the model.
For operators with monotone parameter dependence the spectral bottom of the operator family corresponds
to the configuration where all random variables  are set to one of the extremal values. Furthermore, in the monotone situation the band structure of the spectrum can be analyzed 	
using rather basic bracketing arguments,
see e.g.{} \cite{KirschSS-98a}. It is consistent with these elementary examples of the advantages of monotonicity
that there is a rather good understanding of typical energy/disorder regimes where  monotone models exhibit
localization of waves, see the monographs and survey articles
\cite{Stollmann-01,KirschM-07,Veselic-07b,Kirsch-08}.
\par
If the dependence of the operator, respectively the quadratic form, on the  random variables  is not monotone,
already the identification of the spectral minimum is a highly non-trivial question, see e.g.{}
\cite{BakerLS-08,KloppN-09a}.
For more intricate properties, like the regularity of the density of states
or the analysis of spectral fluctuation boundaries, the difference between monotone
and non-monotone models is even more striking.
\par
Nevertheless there has been a continuous effort to bring the understanding of models
with non-monotone dependence on the randomness to a similar level as the one for monotone models.
This  includes alloy-type Schr\"odinger operators with single-site potentials of changing sign,
see e.g.{} \cite{Klopp-95a,Stolz-00,Klopp-02c,Veselic-02a,HislopK-02,KostrykinV-06, KloppN-09a}, and their lattice counterparts 
cf.~e.g.{} \cite{Veselic-10a,ElgartTV-10,Veselic-10b,TautenhahnV-10,ElgartTV-11,Krueger}.
Electromagnetic Schr\"odinger operators with random magnetic field
 \cite{Ueki-94,Ueki-00,HislopK-02,KloppNNN-03,Ueki-08, Bourgain-09,ErdoesH-a,ErdoesH-c,ErdoesH-b}, 
Laplace-Beltrami operators with random metrics \cite{LenzPV-04,LenzPPV-08,LenzPPV-09},
as well as the random displacement model, cf.{} e.g.{} \cite{Klopp-93,BakerLS-08,KloppLNS},
are other examples with non-monotonous parameter dependence.
\par 
Another relevant model without obvious monotonicity is a random potential 
given by a Gaussian stochastic field with sign-changing covariance function, c.f.{} \cite{HupferLMW-01a,Ueki-04,Veselic-11}.
\par
The methods which have been developed for the discrete alloy-type potentials with sign-indefinite single-site potentials
presented here could be applied to other non-monotone models, as well.
In fact the results of \cite{Krueger} apply to a much larger class of 
lattice Schr\"odinger operator 
with non-monotone, correlated random potential. 
%
%
%
%
%
%
\section{Discrete Anderson models with general alloy-type potential}
\label{sec:model}
To define the random operators we will be looking at, we first introduce the corresponding Hilbert and probability spaces.
Let $d \geq 1$. For $x \in \ZZ^d$ we denote by $\lvert x \rvert_1 = \sum_{i=1}^d \lvert x_i \rvert$ and $\lvert x \rvert_\infty = \max\{\lvert x_1 \rvert,\ldots, \lvert x_d \rvert\}$
the $\ell^1$ and $\ell^\infty$ norms on $\ZZ^d$. 
For $\Gamma \subset \ZZ^d$ we introduce the Hilbert space $\ell^2 (\Gamma) = \{\psi : \Gamma \to \CC : \sum_{k \in \Gamma} \lvert \psi (k) \rvert^2 < \infty\}$ 
with inner product $\sprod{\phi}{\psi} = \sum_{k \in \Gamma} \overline{\phi(k)} \psi (k)$. 
\par
For each $\Gamma\subset \ZZ^d$ we introduce a probability space 
$(\Omega_\Gamma,\cA_\Gamma, \PP_\Gamma) $. Here $\Omega_\Gamma$ is the 
product $\Omega_\Gamma:= \bigtimes_{k \in \Gamma} \RR, $ $\cA_\Gamma$ is the 
associated product sigma algebra generated by cylinder sets, 
and $\PP_\Gamma(\drm \omega) := \prod_{k \in \Gamma} \mu(\drm \omega_k)$ the product measure,
with $\mu$ a probability measure on $\RR$ with bounded support.
The mathe\-matical expectation  with respect to $\PP_\Gamma$ is denoted by 
$\EE_\Gamma$. Note that the projections $\Omega \ni \omega =\{\omega_k\}_{k\in \Gamma} \mapsto \omega_j$, $j \in \Gamma$
give rise to a collection of independent identically distributed (i.i.d.) bounded real random variables.
If $\Gamma=\ZZ^d$ we will suppress the subscript $\Gamma$ in $\Omega_\Gamma, \PP_\Gamma$ and $\EE_\Gamma$.
\par
On $\ell^2 (\ZZ^d)$ we consider the discrete random Schr\"o{}dinger operator
\begin{equation} \label{eq:hamiltonian}
 H_\omega := -\Delta + \lambda V_\omega , \quad \omega \in \Omega , \quad \lambda > 0 .
\end{equation}
Here $\Delta, V_\omega : \ell^2 \left(\ZZ^d\right) \to \ell^2 \left(\ZZ^d\right)$ denote the discrete Laplace 
and  a random multiplication operator defined by
\begin{equation*}
\left(\Delta \psi \right) (x) := \sum_{\abs{e}_1 = 1} \psi (x+e) \quad \mbox{and} \quad
\left( V_\omega \psi \right) (x) :=  V_\omega (x) \psi (x) .
\end{equation*}
The parameter $\lambda$ models the strength of the disorder and $\omega$ denotes the randomness.
It enters the potential in the following way. Let the \emph{single-site potential} $u : \ZZ^d \to \RR$ be a function in $\ell^1 (\ZZ^d ; \RR)$. We assume that the random
potential $V_\omega$ has an \emph {alloy-type structure}, i.e.\ the potential value
\begin{equation*}
V_\omega (x) := \sum_{k \in \ZZ^d} \omega_k u (x-k)
\end{equation*}
at a lattice site $x \in \ZZ^d$ is a linear combination of the i.i.d. random
variables $\omega_k$, $k\in\ZZ^d$, with coefficients provided by the single-site
potential. We assume (without loss of generality) that $0 \in \supp u$. The family of operators $H_\omega$, $\omega \in \Omega$, in Eq. \eqref{eq:hamiltonian} is called \emph{discrete alloy-type model}.
\par
Notice that the single-site potential $u$ may change its sign. As a consequence the quadratic form associated to $H_\omega$ does not necessarily depend in a monotone way on the random parameters $\omega_k$, $k \in \ZZ^d$. This is in sharp contrast to the properties of the \emph{standard Anderson model} which corresponds to the choice of the single-site potential $u =\delta_0$. Here 
\[
 \delta_k(j)  =\begin{cases}
                1 \text{ if } k=j,\\
                0 \text{ else}, 
               \end{cases}
\]
denotes the Dirac function.
%
%
%
%
%
%
\section{Localization properties} \label{sec:loc_properties}
We present several properties related to localization. They concern on the one hand several mathematical signatures of localization, 
and on the other estimates on the average of resolvents and number of eigenvalues in intervals of finite volume systems, 
which are instrumental in the arguments leading to localization. 
They are well established for the standard Anderson model on $\ZZ^d$, 
see e.g.\ \cite{Stollmann-01,GerminetK-04,Kirsch-08} and the references therein. 
\begin{definition}[Dynamical localization]
A selfadjoint operator $H$ on $\ell^2 (\ZZ^d)$ is said to 
exhibit \emph{dynamical localization} in the interval $I\subset \RR$, 
if for every $x \in \ZZ^d$ and $p \geq 1$ we have
\[
 \sup_{t \in \RR} \Bigl( \sum_{n \in \ZZ^d} \bigl(1+\lvert n \rvert_\infty\bigr)^p 
\bigl\lvert \langle \delta_n , \euler^{- \i t H} \chi_I(H)\delta_x \rangle \bigr\rvert^2 \Bigr) < \infty .
\]
\end{definition}
\begin{definition}[Spectral and exponential localization]
 Let $I \subset \RR$. A selfadjoint operator $H : \ell^2 (\ZZ^d) \to \ell^2 (\ZZ^d)$ is said to exhibit \emph{exponential localization in $I$}, 
if the spectrum of $H$ in $I$ is exclusively of pure point type, i.e.\ $\sigma_{\rm c} (H) \cap I = \emptyset$, 
and the eigenfunctions of $H$ corresponding to the eigenvalues in $I$ decay exponentially. If $I=\RR$, we say that
$H$ exhibits \emph{exponential localization}.
\end{definition}
A family of operators $(H_\omega)_\omega$ indexed by elements of a probability space $(\Omega, \PP)$ is said to exhibit {dynamical/exponential localization} in the interval $I\subset \RR$ if the corresponding property holds for $H_\omega$ for almost all $\omega \in \Omega$. 
\par
Several important properties of random Hamiltonians are defined in terms of 
restrictions to a finite system size. We review them next.
Let $\Gamma \subset \ZZ^d$. We define $\Pro_{\Gamma} : \ell^2 (\ZZ^d) \to \ell^2 (\Gamma)$ by
\[
 \Pro_{\Gamma} \psi := \sum_{k \in \Gamma} \psi (k) \delta_k ,
\]
where here $\delta_k$ is the Dirac function in $\ell^2 (\Gamma)$. The restricted operators $\Delta_\Gamma, V_{\omega,\Gamma}, H_{\omega,\Gamma}:\ell^2 (\Gamma) \to \ell^2 (\Gamma)$ are defined by
\[
\Delta_\Gamma := \Pro_\Gamma \Delta \Pro_\Gamma^\ast, \quad V_{\omega,\Gamma} := \Pro_\Gamma V_\omega \Pro_\Gamma^\ast, \quad 
 H_{\omega,\Gamma} := \Pro_\Gamma H_\omega \Pro_\Gamma^\ast = -\Delta_\Gamma + V_{\omega,\Gamma} .
\]
For $z \in \CC \setminus \sigma (H_{\omega,\Gamma})$ we define the corresponding \emph{resolvent} $G_{\omega,\Gamma} (z):= (H_{\omega,\Gamma} - z)^{-1}$ and the \emph{Green function}
\begin{equation*} \label{eq:greens}
G_{\omega,\Gamma} (z;x,y) := \sprod{\delta_x}{(H_{\omega,\Gamma} - z)^{-1}\delta_y}, \quad x,y \in \ZZ^d. 
\end{equation*}
If $\Gamma =\ZZ^d$ we drop the subscript $\Gamma$ in $H_{\omega,\Gamma}, G_{\omega,\Gamma} (z)$ and $G_{\omega,\Gamma} (z;x,y)$. If $\Lambda \subset \ZZ^d$ is finite, $\lvert \Lambda \rvert$ denotes the number of elements of $\Lambda$.
We will use the notation $\RR^+ := \{x \in \RR : x > 0\}$.
\begin{definition}[Decay of fractional moments of the Green function] \label{def:fmb}
There exist constants $s\in (0,1) $ and $A, \gamma \in \RR^+$ 
such that for all $\Gamma \subset \ZZ^d$, $z \in \CC \setminus \RR$ and $x,y \in \Gamma$ we have
\begin{equation*} 
\mathbb{E} \bigl\{\lvert G_{\omega , \Gamma} (z;x,y)\rvert^{s} \bigr\}\leq A \euler^{-\gamma|x-y|_\infty} .
\end{equation*}
\end{definition}
For $x \in \ZZ^d$ and $L>0$, we denote by $\Lambda_{L,x} = \{ k \in \ZZ^d : \lvert x-k  \rvert_\infty \leq L \}$ 
the cube of side length $2L+1$ centred at $x$. For the cube centered at zero we use the notation $\Lambda_L = \Lambda_{L,0}$. We also we write $H_{\omega,L}$ instead of $H_{\omega,\Lambda_{L}}$ and $G_{\omega , L} (z)$ and $G_{\omega , L} (z;x,y)$ instead of $G_{\omega , \Lambda_L} (z)$ and $G_{\omega , \Lambda_L} (z;x,y)$.
For $\Lambda \subset \ZZ^d$ we denote by $\partial^{\rm i} \Lambda = \{k \in \Lambda : \# \{j \in \Lambda : |k-j|_1 = 1\} < 2d\}$ 
the interior boundary of $\Lambda$ and by $\partial^{\rm o} \Lambda = \partial^{\rm i} \Lambda^{\rm c}$ the 
exterior boundary of $\Lambda$. Here $\Lambda^{\rm c} = \ZZ^d \setminus \Lambda$ denotes the complement of $\Lambda$. 
\begin{definition}[Wegner estimate]
There are constants $C_{\rm W},L_0 \in \RR^+$, $b\geq 1$, and a
function $f:\RR^+ \to \RR^+$ satisfying $\lim_{\epsilon \searrow 0} f(\epsilon) = 0$
such that we have for any $L\geq L_0$, $E \in \RR$ and $\epsilon\in (0,1)$
\begin{align} \nonumber
\PP \left \{    [E-\epsilon,E+\epsilon] \cap \sigma(H_{\omega,L})\neq \emptyset   \right\}
&\le
\EE \left \{\Tr \big [\chi_{[E-\epsilon,E+\epsilon]}(H_{\omega,L})\big]\right\}
\\ \label{eq:Wegner} 
&\le
C_{\rm W} f(\epsilon) (2L+1)^{bd} .
\end{align}
\end{definition}
In specific applications of a Wegner estimate (for example as an ingredient for the multiscale analysis) 
one needs a specific rate of decay on the function $f$ near zero. 
The original estimate of Wegner \cite{Wegner-81} corresponds to $f(\epsilon) = \epsilon$ 
and $b=1$ and implies the Lipschitz-continuity of the integrated density of states.
In many situations one can establish \eqref{eq:Wegner} with $b=1$ or $b=2$ and
$f(\epsilon)=\epsilon^a$ for $a \in (0,1)$. 
In certain situations variants of the Wegner estimate 
which do not control  only global property $d(E, \sigma(H_{\omega,L})\leq \epsilon$
but also specific coefficients of the resolvent, as well, need to be used, cf.\ Section \ref{sec:appcartan}
or \cite{Bourgain-09,Krueger}.

\begin{definition}
Let $m,L > 0$, $x \in \ZZ^d$ and $E \in
\RR$. A cube $\Lambda_{L,x}$ is called \emph{$(m,E)$-regular} (for a fixed
potential), if $E \not \in \sigma (H_{\Lambda_{L,x}})$ and
\[
 \sup_{w \in \partial^{\rm i} \Lambda_{L,x}} \lvert G_{\Lambda_{L,x}} (E ; x,w) \rvert
\leq \euler^{-m L} .
\]
Otherwise we say that $\Lambda_{L,x}$ is \emph{$(m , E)$-singular}.
\end{definition}
Rather than looking at the fractional moments of the Green function 
one can consider the probability of a box to be regular. 
The decay of these probabilities is closely related to the localization phenomenon.
For simplicity we define one variant of this decay property and restrict ourselves to the case of finitely supported single-site potentials $u$.
\begin{definition}[Probabilistic decay of Green's function]
Let $\Theta := \supp u$ be finite, $I\subset \RR$ be an interval and let $p>d$, $ L_0>1$, $\alpha \in (1,2p/d)$ and $m>0$.
Set $L_{k} =L_{k-1}^\alpha$, for $k \in \NN$. For any $k \in \NN_0$
\[
 \PP \{\forall \, E \in I \text{ either $\Lambda_{L_k,x}$ or $\Lambda_{L_k,y}$ is $(m,E)$-regular} \}\ge 1- L_k^{-2p}
\]
for any $x,y \in \ZZ^d$ with $\lvert x-y\rvert_\infty \geq 2L_k + \diam \Theta + 1$.
\end{definition}
 Here we denote for finite $\Gamma \subset \ZZ^d$ by $\diam \Gamma$ the diameter of $\Gamma$ 
with respect to the supremum norm, i.e.\ $\diam \Gamma = \sup_{x,y\in \Gamma} \lvert x-y \rvert_\infty$.
In the subsequent sections we describe which of these localization properties have been proven for the non-monotone model we are interested in.
\par
The most general result concerning (large disorder) localization for the discrete alloy-type model is \cite{Krueger}. Kr\"uger proves for exponentially decaying single-site potentials dynamical localization in the case of sufficiently large disorder. 
Indeed, this result applies for a class of models including the discrete alloy-type model with exponentially decaying singe-site potential as a special case. Notice that dynamical localization implies spectral localization via the RAGE-Theorem, see e.g.\ \cite{Stolz2010}, but not vice versa as examples in \cite{RioJLS1996} show. The proof of Kr\"uger's result uses the multiscale analysis and is discussed in Section \ref{sec:msa}. 
\par
There are also localization results not using the multiscale analysis but the fractional moment method. In \cite{ElgartTV-10} the authors prove exponential localization in the case of space dimension $d=1$, compactly supported single-site potentials and sufficiently large disorder. 
This result was extended in \cite{ElgartTV-11} to arbitrary space dimension assuming that the single-site potential has fixed sign at the boundary of its support, 
a property which can be assumed without loss of generality in $d=1$. This result is discussed in Section \ref{sec:loc_fmm}.
\par
In Section \ref{sec:averaging} we present certain estimates concerning averages of polynomials and resolvents which are fundamental for the results presented in Sections \ref{sec:loc_fmm} and \ref{sec:msa}.
\par
In Section \ref{sec:Wegner} we discuss certain results on Wegner estimates for the discrete alloy-type model \cite{Veselic-10b,PeyerimhoffTV}. 
With the help of such Wegner-estimates one can implement  a proof of localization via the multiscale analysis in the regime where an appropriate initial length scale estimate is available.
\par
In the following Section \ref{sec:spectrum} we show that the almost sure spectrum of the discrete alloy-type model is an interval.
%
%
%
%
%
%
%
%
\section{The spectrum}\label{sec:spectrum}
Before studying the properties of the spectral measure under the 
Lebesgue decomposition, one wants to understand 
basic features of the set $\Sigma\subset \RR$ 
which coincides with the spectrum of $H_\omega$ almost surely. 
This concerns in particular 
the infimum and supremum of the spectrum and internal spectral edges (if any).
For the standard Anderson model there is a nice formula for the spectrum:
\[
 \Sigma= [-2d, 2d] + \supp \mu
\]
where $[-2d, 2d]$ is the spectrum of the free Laplacian $\Delta$
and $\supp \mu$ the spectrum of the multiplication operator given by the random potential.
Related descriptions of $\Sigma$ for the (continuum) alloy-type model have been studied
in  \cite{KirschM-82b}. In particular, a description of $\Sigma$ 
in terms of admissible potentials was given. In many cases this class
consists of an appropriate family of periodic potentials.
Let us quote a specific result from \cite{KirschSS-98a}: If  $S:=\supp \mu$ is a bounded  interval
and the single-site potential $u$ is non negative, then 
\[
 \Sigma=\bigcup_{\kappa \in S} \sigma \left(-\Delta + \kappa \sum_{n\in \ZZ^d} u(\cdot-n)\right).
\]
The proof of this equality uses that $u$ has fixed sign and is thus not applicable in our case.
Our result about the set $\Sigma$ is 

\begin{theorem}
Let $H_{\omega}$  be a  discrete random Schr\"odinger operators as in \eqref{eq:hamiltonian}.
If $S:=\supp \mu$ is bounded and connected,  
then the spectrum of $H_{\omega}$ is almost surely an interval.
\end{theorem}
\begin{proof}
Denote by $\Sigma$ the almost sure spectrum of $H_{\omega}$. It is well know that 
\[
 \Sigma= \overline{\bigcup_{\omega\in S^{\ZZ^d}}\sigma(H_\omega)}
\]
cf.\ the discussion in Section 6 of \cite{Kirsch-89a}. 
Now we assume $0 \in S$ without loss of generality. 
For $\omega \in S^{\ZZ^d}$ and $t \in [0,1]$
denote $t \cdot \omega = (t \omega_x)_{x \in \ZZ^d}$ the scaled configuration.
We then have that $\sigma(H_{0\cdot \omega}) = [-2d, 2d]$.
Fix now $\omega \in S^{\ZZ^d}$. We show that
$$
\Sigma_\omega = \bigcup_{t \in [0,1]} \sigma(H_{t \cdot \omega})
 $$
is an interval. For this purpose note that 
$\sigma (H_{0\cdot \omega}) = [-2d,2d]$ and that 
the spectral maximum and spectral minimum 
\[
 [0,1] \ni t \mapsto \max \sigma (H_{t \cdot \omega}), \quad 
 [0,1] \ni t \mapsto \min \sigma (H_{t \cdot \omega})
\]
are continuous functions of $t$. This follows by the min-max-principle. 
Thus 
$$
\Sigma_\omega = \Bigl[\min_{t \in [0,1]} \min \sigma (H_{t\cdot \omega}) , \max_{t \in [0,1]} \max \sigma (H_{t\cdot \omega})\Bigr] \supset [-2d,2d] .
$$
Choose now $\lambda \in \bigcup_{\omega\in S^{\ZZ^d}}\sigma(H_\omega)$. Then there exists $\tilde \omega \in S^{\ZZ^d}$ such that $\lambda \in \sigma (H_{\tilde \omega}) \subset \Sigma_{\tilde \omega}$. Since the latter set is an interval containing $[-2d,2d]$, it follows that $\Sigma$ is an interval.
\end{proof}
%
%
%
%
%
%
\section{Averaging of determinants and resolvents} \label{sec:averaging}
In the energy regime where localization holds, eigenvalues are sensitive to fluctuations of the random potential.
In particular, the mathematical expectation leads to a regularization of the finite volume eigenvalue counting function.
Likewise, (appropriate) averages of the resolvent enjoy boundedness properties which are impossible to hold for 
resolvent associated to individual realizations of the random potential.
In this section we discuss bounds of the type indicated above and in the following sections their relation to localization proofs. 
\par
We start with a well known weak $L^1$-bound formulated in Lemma~\ref{lemma:monotone} (see e.g.\ \cite[Proposition 3.1]{AizenmanENSS-06} for a more general result), which can be used to obtain bounds on certain averages of resolvents in the monotone case, i.e.\ where the single-site potential $u$ is non-negative. 
Recall, a densely defined operator $T$ on some Hilbert space $\mathcal{H}$ with inner product $\langle \cdot , \cdot \rangle_{\mathcal{H}}$ is called \emph{dissipative} if $\im \langle x,Tx \rangle_{\mathcal{H}} \geq 0$ for all $x \in D(T)$.

\begin{lemma} \label{lemma:monotone}
Let $A \in \CC^{n \times n}$ be a dissipative matrix, $V \in \RR^{n \times n}$ diagonal and strictly positive definite and $M_1 , M_2 \in \CC^{n \times n}$ be arbitrary matrices. Then there exists a constant $C_{\rm W}$ (independent of $A$, $V$, $M_1$, $M_2$ and $n$), such that
\[
\cL \bigl\{ r \in \RR : \lVert M_1 (A + r V)^{-1} M_2 \rVert_{\rm HS} > t \bigr\} \leq C_{\rm W} \lVert M_1 V^{-1/2} \rVert_{\rm HS} \lVert M_2 V^{-1/2} \rVert_{\rm HS} \frac{1}{t} .
\]
Here, $\cL$ denotes the Lebesgue-measure and $\lVert \cdot \rVert_{\rm HS}$ the Hilbert Schmidt norm.
\end{lemma}
If the single-site potential $u$ has fixed sign (and compact support) Lemma \ref{lemma:monotone} is applicable and yields (together with a decoupling argument) the decay of fractional moments of the Green function. A generalization of Lemma \ref{lemma:monotone} also applies to the continuous alloy-type model on $L^2 (\RR^d)$ to obtain bounds on fractional moments of the Green function, see \cite{AizenmanENSS-06}.

Since we allow the single-site potential to change its sign, we want to get rid of the positivity assumption on the operator $V$. The first observation is, that if one considers averages of determinants, then the definiteness of $V$ plays no longer a role.
\subsection{Estimates on polynomials and resolvents}
In this section, we discuss estimates as used by Elgart, Tautenhahn, and Veseli\'c in \cite{ElgartTV-10} and \cite{ElgartTV-11}.
\begin{lemma}[\cite{ElgartTV-10}] \label{lemma:det}
 Let $n \in \NN$, $P(x) = x^n + \dots$ a polynomial of degree $n$, 
 and  $0 \leq \rho \in L^1(\RR) \cap L^\infty (\RR)$ and $s \in (0,1)$.
 Then we have
 \be\label{eq:etv1}
  \int_{\RR} \frac{1}{|P(x)|^{s/n}} \rho(x) dx 
   \leq \Vert \rho \Vert_{L^1}^{1-s}
    \Vert\rho\Vert_{\infty}^{s} \frac{2^{s} s^{-s}}{1-s}.
 \ee
\end{lemma}
In \cite{ElgartTV-10}, this result is stated as
\begin{align*}
 \int_{\RR} \abs{\det (A + rV)}^{-s/n} \rho (r) \drm r
&\leq \abs{\det V}^{-s/n} \Vert \rho \Vert_{L^1}^{1-s} \Vert\rho\Vert_{\infty}^{s} \frac{2^{s} s^{-s}}{1-s}  \\
&\leq \abs{\det V}^{-s/n}\Bigl( \lambda^{-s} \Vert\rho\Vert_{L^1} + \frac{2 \lambda^{1-s}}{1-s} \Vert\rho\Vert_\infty  \Bigr) 
\end{align*}
holds for $\lambda > 0$, $A \in \CC^{n \times n}$,
and $V \in \CC^{n \times n}$ invertible. Since $r \mapsto \frac{1}{\det(V)} \det(A + r V)$
is a monic polynomial of degree $n$, \eqref{eq:etv1} implies this statement.
For the converse, use that $P(x) = x^n + \sum_{j=0}^{n-1} \alpha_j x^j$
can be rewritten in terms of the companion matrix as
\[
 P(x) = \det(x + A),\quad A = \begin{pmatrix} 0 & 0 & \dots & 0 & - \alpha_0 \\
 1 & 0 & \dots & 0 & - \alpha_1 \\
 \vdots & \vdots & \ddots & \vdots & \vdots \\
 0 & 0 & \dots &0 & -\alpha_{n-1} \\
 0 & 0 & \dots & 1 & - \alpha_n\end{pmatrix}.
\]
The given form of Lemma~\ref{lemma:det} has the advantage that
its relation to P\'olya's inequality becomes more apparent. Namely, that
for any polynomial $P(x) = x^n + \dots$ of degree $n$, we have
\be
 |\{x\in \RR :\quad |P(x)| \leq \alpha\}| \leq 4 \left(\frac{\alpha}{2}\right)^{\frac{1}{n}}
\ee
for $\alpha > 0$.
\par
In $d = 1$ the bound from Lemma \ref{lemma:det} is precisely what is needed to show the boundedness of averaged fractional powers of the Green function, since certain matrix elements of the Green function can be represented as an inverse of a determinant of the above type. More precisely, if $d = 1$ and $\supp u = \{0,1,\ldots , n-1\}$ then we have for all $x \in \ZZ$ and $z \in \CC \setminus \RR$
\[
 \lvert G_\omega (z ; x,x+n-1) \rvert = \frac{1}{\lvert \det (A + \omega_x \lambda V) \rvert} ,
\]
where $V \in \RR^{n \times n}$ is diagonal with diagonal elements $u (k-x)$, $k = x , \ldots , x+n-1$, and where $A \in \CC^{n \times n}$ is independent of $\omega_x$. By Lemma~\ref{lemma:det} one obtains a bound on the expectation of an averaged fractional power on certain Green's function elements, which is sufficient to start the proof of localization via the fractional moment method. 
See \cite{ElgartTV-10} for details.
\par
Let us turn to the higher dimensional case.
For $B \in \CC^{n\times n}$ we have
\[
 \|B^{-1}\| \leq \frac{\|B\|^{n-1}}{|\det(B)|} .
\]
Thus one can use \eqref{eq:etv1} to obtain bounds on the inverse of
$A  + r V$. More precisely, we have for $s \in (0,1)$, $R > 0$ and $A,V$ as above the estimate 
\begin{equation} \label{eq:average_norm}
\int_{-R}^R \bigl\Vert (A+rV)^{-1} \bigr\Vert^{s/n} {\rm d}r \leq \frac{2R^{1-s} (\Vert A \Vert + R \Vert V \Vert)^{s(n-1)/n}}{s^s (1-s) \abs{\det V}^{s/n}} .
\end{equation}
In $d>1$ Lemma \ref{lemma:det} is no longer applicable, but the estimate \eqref{eq:average_norm} is. However, the problem with the estimate \eqref{eq:average_norm} is that the upper bound depends on the background operator $A$. Note that $A$ arises by a Schur-complement formula and has a complicated dependence on the randomness. 
For this reason we have to assume additional properties on the single-site potential $u$ (i.e.\ that $u$ has fixed sign on the boundary of its (bounded) support) to show bounds on averaged fractional powers of the Green function, 
see Section~\ref{sec:loc_fmm} or \cite{ElgartTV-11} for details.
\subsection{Cartan estimates}
We have now discussed the facts from complex analysis on which
\cite{ElgartTV-10} and \cite{ElgartTV-11} are based. Let us now
discuss the ones used in \cite{Krueger} and Bourgain
\cite{Bourgain-09}, which are based on ideas from quasi-periodic Schroedinger operators see Bourgain, Goldstein and Schlag \cite{BourgainGS-02,Bourgain-04,Bourgain-07}
\par
Like the estimates discussed in the first part of this section, they
control the size of the set where an analytic function is
small. The main difference is that instead of requiring as
input that the analytic function is a polynomial of
small degree, one assumes that the analytic function is
not small at a single point. Here is the simplest form
of such an estimate.

\begin{theorem}\label{thm:cartan1}
 Let $f$ be an analytic function on the disc of radius
 $2 \euler$ satisfying
 \be \nonumber
  \sup_{|z| < 2 \euler} |f(z)| \leq 1,\quad |f(0)| \geq \varepsilon.
 \ee
 Then for $s > 0$, we have
 \be \nonumber
  |\{x \in [-1,1]:\quad |f(x)|\leq \euler^{-s}\}|
   \leq 30 \euler^3 \exp\left(-\frac{s}{\log(\varepsilon^{-1})}\right).
 \ee
\end{theorem}

\begin{proof}
 For a proof see Theorem~11.3.4 in \cite{Levin-96}. See also Theorem~10.2 in \cite{Krueger} for the deduction of
 this statement.
\end{proof}

This form of the Cartan estimate is not sufficient for the application
to random Schr\"odinger operators as done in \cite{Bourgain-09}
and \cite{Krueger}. For these results, one needs to apply the Cartan
estimate to functions depending on many variables, and thus needs
an estimate that is well behaved in the number of variables.
Such an estimate was first proven by Nazarov, Sodin, and Volberg \cite{NazarovSV-03}.
Unfortunately, they work on balls and not on poly-disks as
necessary for our applications. The result for poly-disks
was proven by Bourgain in \cite{Bourgain-09}.
We will state here a formulation which was given in  \cite{Krueger}.

\begin{theorem}[Theorem~10.1, \cite{Krueger}]\label{thm:cartan2}
Denote by $\DD_r$ the disc of radius $r$ in $\CC$. 
 Let $f: (\DD_{2\euler})^n\to\CC$ be an analytic function satisfying
 \be \nonumber
  \|f\|_{L^{\infty}( (\DD_{2\euler})^n )} \leq 1,\quad |f(0)| \geq \varepsilon.
 \ee
 Then for $s > 0$, we have
 \be \nonumber
  \frac{|\{x \in [-1,1]^n:\quad |f(x)|\leq \euler^{-s}\}|}{2^n}
   \leq 30 \euler^3 n  \exp\left(-\frac{s}{\log(\varepsilon^{-1})}\right).
 \ee
\end{theorem}

Note that the dimension dependence is $n$. 
The proof is based on a clever change to spherical coordinates.
We give an exposition on how to apply this to 
Schr\"odinger operators in Section~\ref{sec:appcartan}.
%
%
%
%
%
%
\section{Wegner estimates} \label{sec:Wegner}

This section is concerned with averaging of spectral projections.
For any $\omega \in \Omega$ and $L\in \NN$ the restriction $H_{\omega,L}$ 
is a selfadjoint finite rank operator. 
In particular its spectrum consists entirely of real eigenvalues 
$E(\omega,L,1)\le E(\omega,L,2)\le \dots \le E(\omega,L,\sharp \Lambda_L)$ 
counted including multiplicities. 
Wegner estimates \cite{Wegner-81} are bounds on the expected number of eigenvalues 
of the finite box Hamiltonians $H_{\omega,L}$ in a compact energy interval $I = [E-\epsilon, E + \epsilon]$. 
They can be used as an ingredient for a localization proof via the multiscale analysis. 
The symbol $\chi_I (H_{\omega , L})$ denotes the spectral projection onto 
$I$ with respect to the operator $H_{\omega , L}$.
\begin{theorem}[\cite{Veselic-10b},\cite{PeyerimhoffTV}] \label{thm:Wegner_d}
Assume that $\mu$ has a density $\rho$ of finite total variation
and $u$ is not identically zero.
\begin{enumerate}[(a)]
 \item 
Assume that the single-site potential $u$ has support in $[0,n]^d \cap \ZZ^d$.
Then there exists a constant $c_u$ depending only on $u$ such that for any
$L\in \NN$, $E \in \RR$ and $\epsilon>0$ we have
\[
\EE \left \{\Tr \chi_{[E-\epsilon,E+\epsilon]}(H_{\omega,L})\right\}
\le
c_u \,  \lvert \supp u \rvert \,   \|\rho\|_{\rm BV}   \ \epsilon\, (2L+n)^{d\cdot (n+1)}
\]
\item 
Assume $\bar u = \sum_{k \in \ZZ^d} u(k) \neq 0$ and that $\rho$ has compact support. Then we have
for any $L\in \NN$, $E \in \RR$ and $\epsilon>0$
\[
\EE \left \{\Tr \chi_{[E-\epsilon,E+\epsilon]}(H_{\omega,L})\right\}
\le
\frac{8}{\bar u}  \,  \min\bigl((2L)^d,\lvert \supp u \rvert\bigr)  \, \|\rho\|_{\rm BV}  \ \epsilon\, (2L+m)^{d} ,
\]
where $m \in \NN$ is such that $\sum_{\lVert k \rVert \geq m} \lvert u(k) \rvert \leq \lvert \bar u / 2 \rvert$.
\item 
Assume there are constants $C,\alpha \in \RR^+$ such that 
$\lvert u(k) \rvert \le C e^{- \alpha \Vert k \Vert_1}$ for all $k \in \ZZ^d$, and that $\rho$ has bounded support.
Then there exists $c_u>0$ and $I_0 \in \NN_0^d$ both depending only on $u$ 
such that for any $L \in \NN$, $E \in \RR$ and $\epsilon > 0$ we have
\[
 \EE \bigl \{\Tr \chi_{[E-\epsilon , E+\epsilon]} (H_{\omega,L}) \bigr\}\le
c_u \lVert \rho \rVert_{\rm BV} \, \epsilon \, (2L+1)^{2d + \lvert I_0 \rvert} .
\]
\end{enumerate}
\end{theorem}
In the case where the support of $u$ is compact, part (b) of Theorem \ref{thm:Wegner_d} has an important corollary.
\begin{corollary}[\cite{Veselic-10b}] \label{cor:dos}
Assume $\bar u \neq 0$ and $ \supp u \subset [0,n]\cap \ZZ^d$.
Then we have for any $L\in \NN$, $E \in \RR$ and $\epsilon>0$
\[
\EE \left \{\Tr \chi_{[E-\epsilon,E+\epsilon]}(H_{\omega,L})\right\}
\le
\frac{4}{\bar u}  \, \rank u \,   \|\rho\|_{BV} \ \epsilon\, (2L+n)^{d} .
\]
In particular, the function  $\RR \ni E \to \EE \left \{\Tr \big [\chi_{(-\infty,E]}(H_{\omega,L})\big]\right\}$
is  Lipschitz continuous.
\end{corollary}
\begin{remark}
\begin{enumerate}[(i)]
 \item Note, that apart of the exponential decay condition on $u$, 
  Theorem \ref{thm:Wegner_d} (c) gives a Wegner estimate for the 
  discrete alloy-type model without any further assumption on the single-site potential. 
  In particular, $u$ may change its sign arbitrarily.
 \item The proof of Theorem \ref{thm:Wegner_d} is (roughly speaking) based on a transformation of the probability space to recover monotonicity. With other words, once you find a finite linear combination of translated single-site potentials which is positive, then monotone spectral averaging leads to a Wegner estimate, see \cite{KostrykinV-06,Veselic-10a,Veselic-10b,PeyerimhoffTV} where this approach is used.
\end{enumerate}
\end{remark}
\begin{remark}[Continuous alloy-type model]
The \emph{alloy-type model} is the Schr\"o{}dinger operator $H_\omega = -\Delta + V_0 + V_\omega$ on $L^2 (\RR^d)$, where $\Delta$ is the Laplacian on $\RR^d$, $V_0$ a $\ZZ^d$ periodic potential, and $V_\omega$ given by 
\[
 V_\omega (x) = \sum_{k \in \ZZ^d} \omega_k U(x-k)
\]
with $U : \RR^d \to \RR$ a single-site potential. 
It is assumed that $V_0$ and $V_\omega$ are infinitesimally bounded with respect to $\Delta$, 
with constants uniformly bounded in $\omega \in \Omega$.
We will be concerned with the case that  the distribution $\mu$ has a density $\rho$ of finite total variation. 
and  $U$ is a \emph{generalized step-function}, i.e.
\[
U (x) = \sum_{k \in \ZZ^d} u (k) w(x-k).
\]
Here $L_{\rm c}^p (\RR^d) \ni w \geq \kappa \chi_{(-1/2,1/2)^d}$ with $\kappa > 0$, $p = 2$ for $d \leq 3$ and $p>d/2$ for $d \geq 4$, and $u \in \ell^1 (\ZZ^d ; \RR)$ the discrete single-site potential. 
\par
In \cite{PeyerimhoffTV} a Wegner estimate similar to part (c) in Theorem \ref{thm:Wegner_d} 
is proven for the (continuous) alloy-type model. More precisely, assume that $U$ is a generalized step-function and there are constants $C,\alpha \in \RR^+$ such that $\lvert u (k) \rvert \leq C \euler^{-\alpha \lVert k \rVert_1}$. Then there exists $c_U > 0$ and $I_0 \in \NN_0^d$ both depending only on $U$ such that for any $L \in \NN$ and any bounded interval $I:= [E_1,E_2] \subset \RR$ we have
 \[
\EE \bigl \{\Tr \chi_I (H_{\omega,L}) \bigr\}
\le
\euler^{E_2} c_{U} \lVert \rho \rVert_{\rm Var} \lvert I \rvert (2l+1)^{2d + \lvert I_0 \rvert} .
\]
Here $H_{\omega , L}$ denotes the restriction of $H_\omega$ to the cube $(-L,L)^d \subset \RR^d$ with either Dirichlet or periodic boundary conditions. The stated Wegner estimate is valid for both types of boundary conditions. 
\par
A drawback of this results is that they fail if $u$ is not a generalized step function.
Contrary to this, the papers \cite{Klopp-95a,HislopK-02} obtain Wegner estimates for energies in a neighborhood of the infimum of the spectrum 
which are valid for arbitrary non-vanishing single-site potentials $u \in C_{\rm c}(\RR^d)$ and coupling constants whose distribution has a piecewise absolutely continuous density. 
\end{remark}
\begin{remark}[Localization]
 Notice that the Wegner estimates from Theorem~\ref{thm:Wegner_d} are valid on the whole energy axis. Therefore, one can prove localization via multiscale analysis \cite{FroehlichS-83,DreifusK-89} in any energy region where an initial length scale estimate holds. If the single-site potential does not have compact support one has to use a modified version of the multiscale analysis \cite{KirschSS-98b}.
\end{remark}
%
%
%
%
%
%
%
%
\section{Localization via fractional moment method} \label{sec:loc_fmm}
\subsection{Boundedness of fractional moments}
The lemmata from Section~\ref{sec:averaging} can be used to obtain bounds on averaged fractional powers of the Green function. 
One possible approach to overcome the problems arising 
because of the lack of monotonicity is to use a special transformation of the probability space to recover some monotonicity which makes Lemma \ref{lemma:monotone} applicable. This was done in \cite[Appendix]{ElgartTV-11} to obtain the following Theorem.
\begin{theorem}[\cite{ElgartTV-11}] \label{theorem:apriori_average}
Assume  \begin{enumerate}[(i)]
\item The measure $\mu$ has a density $\rho$ in the Sobolev space $ W^{1,1} (\RR)$.
\item  The single-site potential $u$ has compact support and satisfies $\overline u := \sum_{k \in \ZZ^d} u(k) \neq 0$.
\end{enumerate}
 Let further $\Lambda \subset \ZZ^d$ finite and $s \in (0,1)$. Then we have for all $x,y \in \Lambda$ and $z \in \CC \setminus \RR$
\[
\mathbb{E} \Bigl\{ \bigl\lvert G_{\omega , \Lambda} (z;x,y) \bigr\rvert^s \Bigr\}  
\leq
\frac{1}{1-s} 
\left(\frac{2\lVert \rho' \rVert_{L^1} C}{\overline{u}}\right)^s 
\frac{1}{\lambda^s}
\]
where $C$ is a constants depending only on $u$.
\end{theorem}
The disadvantage of Theorem \ref{theorem:apriori_average} is that it is non-local in the sense that we have to average with respect of the entire disorder present in the model. 
At the moment we do not know how to conclude the decay of fractional moments of the 
Green function (cf.\ Definition \ref{def:fmb}) from the non-local a priori bound in Theorem \ref{theorem:apriori_average}.
A somewhat stronger condition, however, is sufficient to ensure decay of fractional moments.
We will review this result next.
\begin{assumption} \label{ass:monotone}
Assume that
\begin{enumerate}
 \item 
 the measure $\mu$ has a bounded, compactly supported density $\rho \in L^{\infty} (\RR)$.
\item 
$\Theta := \supp u$ is finite and $u (k) > 0$ for $k \in \partial^{\rm i} \Theta$.
\end{enumerate}
\end{assumption}
Under Assumption \ref{ass:monotone} and with the help of Ineq. \eqref{eq:average_norm} it is possible to prove a local a priori bound, which is applicable to conclude the decay of fractional moments of the Green function, see Section~\ref{sec:fmb_decay}. 
Let us first introduce some more notation.
For $x \in \ZZ^d$ we denote by $\cN (x) = \{k \in \ZZ^d : |x-k|_1 = 1\}$ the neighborhood of $x$. For $\Lambda \subset \ZZ^d$ we also define $\Lambda_x = \Lambda + x = \{k \in \ZZ^d : k-x \in \Lambda\}$. 
\begin{lemma}[\cite{ElgartTV-11}] \label{lemma:bounded}
Let Assumption~\ref{ass:monotone} be satisfied, $\Gamma \subset \ZZ^d$, $m > 0$ and $s \in (0,1)$.
\begin{enumerate}[(a)]
 \item Then there is a constant $C_s$, depending only on $d$, $\rho$, $u$, $m$ and $s$, such that for all $z \in \CC \setminus \RR$ with $|z| \leq m$, all $x,y \in \Gamma$ and all $b_x,b_y \in \ZZ^d$ with $x \in \Theta_{b_x}$ and $y \in \Theta_{b_y}$
\[
\EE_{N} \Bigl\{ \bigl\lvert G_{\omega,\Gamma} (z;x,y) \bigr\rvert^{s/(2|\Theta|)} \Bigr\} \leq C_s \Xi_s (\lambda),
\]
where $\Xi_s (\lambda) = \max \{ \lambda^{- s/(2 \lvert \Theta \rvert)} , \lambda^{-2s} \}$ and $N = \{b_x , b_y\} \cup \cN (b_x) \cup \cN (b_y)$.
\item Then there is a constant $D_s$, depending only on $d$, $\rho$, $u$ and $s$, such that for all $z \in \CC \setminus \RR$, all $x,y \in \Gamma$ and all $b_x,b_y \in \ZZ^d$ with
\[
x \in \Theta_{b_x} \cap \Gamma \subset \partial^{\rm i} \Theta_{b_x} \quad \text{and} \quad
y \in \Theta_{b_y} \cap \Gamma \subset \partial^{\rm i} \Theta_{b_y}
\]
we have
\[
 \EE_{\{b_x , b_y\}} \Bigl\{ \bigl\lvert G_{\omega,\Gamma} (z;x,y) \bigr\rvert^s \Bigr\} \leq D_s \lambda^{-s} .
\]
\end{enumerate}
\end{lemma}
\subsection{Decay of fractional moments} \label{sec:fmb_decay}
Now we explain how the so called finite volume criterion implies exponential decay of the Green function. 
Together with the a-priori bound from Lemma \ref{lemma:bounded} this gives us Theorem~\ref{thm:result1}. For proofs we refer the reader to \cite{ElgartTV-11}.
\par
To formulate the results of this section we will need the following
notation: Let $\Gamma \subset \ZZ^d$, fix $L \ge \diam \Theta + 2$, let
\[
B = \partial^{\rm i} \Lambda_L,
\]
and define the sets
\[
\hat \Lambda_x = \{ k \in \Gamma : k \in \Theta_b \text{ for some $b \in \Lambda_{L,x}$} \}
\]
and
\begin{equation*} \label{eq:Wx}
 \hat W_x = \{ k \in \Gamma : k \in \Theta_b \text{ for some $b \in B_x$} \} .
\end{equation*}
Recall that for $\Gamma \subset \ZZ^d$ we denote by $\Gamma_x = \Gamma + x = \{k \in \ZZ^d : k-x \in \Gamma\}$ the translate of $\Gamma$. 
Hence $(\Lambda_L)_x=\Lambda_{L,x}$ and 
$\hat W_x$ is the union of translates of $\Theta$ along the sides of $B_x$, restricted to the set $\Gamma$. For $\Gamma \subset \ZZ^d$ we can now introduce the sets
\[
\Lambda_x: = \hat \Lambda_x^+ \cap \Gamma \quad \text{and} \quad W_x = \hat W_x^+ \cap \Gamma
\]
which will play a role in the assertions below.

\begin{theorem}[\cite{ElgartTV-11}, Finite volume criterion]
\label{thm:exp_decay}
Suppose that Assumption \ref{ass:monotone} is satisfied.
Let $\Gamma \subset \ZZ^d$, $z \in \CC \setminus \RR$ with $\lvert z \rvert \leq m$ and $s \in (0,1/3)$. 
Then there exists a constant $B_s$ which depends only on $d$, $\rho$, $u$, $m$, $s$,
such that if the condition
\begin{equation}\label{eq:fin_cond}
b_s(\lambda, L,\Lambda): = \frac{B_s L^{3(d-1)} \Xi_s (\lambda)}{\lambda^{2s/(2\lvert
\Theta \rvert)}} \!\!\! \sum_{w\in\partial^{\rm o} W_x} \!\!\! \mathbb{E}
\bigl\{\lvert G_{\omega , \Lambda\setminus W_x} (z;x,w)\rvert^{s/(2\lvert
\Theta \rvert)}\bigr\}< b
\end{equation}
is satisfied for some $b \in (0,1)$, arbitrary $\Lambda \subset \Gamma$, and all $x\in
\Lambda$, then for all $x,y \in \Gamma$
\begin{equation*} 
\mathbb{E} \bigl\{\lvert G_{\omega,\Gamma} (z;x,y)\rvert^{s/(2\lvert \Theta
\rvert)} \bigr\}\leq A \euler^{-\mu|x-y|_\infty} .
\end{equation*}
Here
\[
A=\frac{C_s \Xi_s (\lambda)}{b} \quad \text{and} \quad
\mu=\frac{\lvert \ln b \rvert}{L+\diam \Theta + 2}\,,
\] 
with $C_s$ as in Lemma \ref{lemma:bounded}.
\end{theorem}
Note that condition \eqref{eq:fin_cond} can be achieved by choosing $\lambda$ sufficiently big
and applying Lemma \ref{lemma:bounded}.
The core of the proof of the theorem is the following Lemma \ref{lemma:bounded}.
\begin{lemma}[\cite{ElgartTV-11}] \label{lemma:iteration1}
Let Assumption \ref{ass:monotone} be satisfied. Let $\Gamma \subset \ZZ^d$, $s \in (0, 1/3)$, $m > 0$ 
and $b_s (\lambda,L,\Lambda)$ be the constant from Theorem \ref{thm:exp_decay}. 
Then we have for all $x,y \in \Gamma$ with $y \not \in \Lambda_x$ and all $z \in \CC \setminus \RR$ with $\lvert z \rvert \leq m$ the bound
\begin{equation*}\label{eq:protobound}
\mathbb{E} \bigl\{\lvert G_{\omega,\Gamma} (z;x,y)\rvert^{\frac{s}{2\lvert
\Theta \rvert}}\bigr\}\leq
\frac{b_s(\lambda, L,\Gamma)}{|\partial^{\rm o} \Lambda_x|} \sum_{r\in\partial^{\rm o}
\Lambda_x} \mathbb{E}\bigl\{ \lvert G_{\omega , \Gamma\setminus
\Lambda_x}(z;r,y)\rvert^{\frac{s}{2\lvert \Theta \rvert}} \bigr\} .
\end{equation*}
\end{lemma}
The combination of Theorem    \ref{thm:exp_decay}     and Lemma  \ref{lemma:bounded}
yield the following result on exponential decay of a fractional moment of the 
Green function under a strong disorder assumption.
\begin{theorem}[\cite{ElgartTV-11}] \label{thm:result1}
Let $\Gamma \subset \ZZ^d$, $s \in (0,1/3)$ and suppose that Assumption \ref{ass:monotone} is satisfied. 
Then for a sufficiently large $\lambda$ there are constants $\mu,A \in \RR^+$, 
depending only on $d$, $\rho$, $u$,  $s$ and $\lambda$, 
such that for all $z \in \CC \setminus \RR$ and all $x,y \in \Gamma$
\begin{equation*} \label{eq:result1}
\mathbb{E} \bigl\{\lvert G_{\omega,\Gamma} (z;x,y)\rvert^{s/(2\lvert \Theta \rvert)}\bigr\}\leq A \euler^{-\mu|x-y|_\infty} .
\end{equation*}
\end{theorem}
\subsection{Localization}
The existing proofs of localization via the fractional moment method
use either the Simon Wolff criterion, see e.g.\
\cite{SimonW-86,AizenmanM-93,AizenmanFSH-01}, or the
RAGE-Theorem, see e.g.\
\cite{Aizenman-94,Graf-94,AizenmanENSS-06}. Neither dynamical nor
spectral localization can be directly inferred from the behavior of
the Green function using the existent methods for the model in Section \ref{sec:model}. 
The reason is that the random variables $V_\omega (x)$, $x \in \ZZ^d$,
are not independent, while the dependence of $H_\omega$ on the
i.i.d. random variables $\omega_k$, $k \in \ZZ^d$, is not
monotone.
\par
We outline a new variant for concluding exponential localization 
from bounds on averaged fractional powers of Green's function (cf.\ Section \ref{sec:fmb_decay}) without using the multiscale
analysis, see \cite{ElgartTV-10, ElgartTV-11} for details. This is done by showing that fractional moment bounds imply the ``typical output'' of the multiscale analysis, i.e.\ the
hypothesis of Theorem 2.3 in \cite{DreifusK-89}. Then one can
conclude localization using existent methods.

The next Proposition states that certain bounds on averaged fractional
moments of Green's function 
imply the hypothesis of Theorem 2.3 in \cite{DreifusK-89} (without applying the induction step of the multiscale analysis).
\begin{proposition}[\cite{ElgartTV-11}] \label{prop:replace-msa}
Let $I \subset \RR$ be a bounded interval and $s \in (0,1)$. Assume the following two statements:
\begin{enumerate}[(i)]
\item There are constants $C,\mu \in \RR^+$ and $L_0 \in \NN_0$
such that 
\[
\EE \bigl\{ \lvert G_{\omega , \Lambda_{L,k}} (E;x,y) \rvert^{s} \bigr\} \leq C \euler^{-\mu \lvert x-y \rvert_\infty}
\]
for all $k \in \ZZ^d$, $L \in \NN$, $x,y \in \Lambda_{L,k}$ with $\lvert x-y \rvert_\infty \geq L_0$, and all $E \in I$.
\item There is a constant $C' \in \RR^+$ such that 
\[
\EE \bigl\{ \lvert G_{\omega , \Lambda_{L,k}} (E+\i \epsilon ;x,x) \rvert^{s} \bigr\} \leq C'
\]
for all $k \in \ZZ^d$, $L \in \NN$, $x \in \Lambda_{L,k}$, $E \in I$ and all $\epsilon\in (0,1]$ .
\end{enumerate}
Then we have for all $L \geq \max\{ 8\ln (8)/\mu , L_0 , -(8/5\mu)\ln (\lvert I \rvert / 2)\}$ and all $x,y \in \ZZ^d$ with $\lvert x-y\rvert_\infty \geq 2L+\diam \Theta + 1$ that
\begin{multline*}
 \PP \{\forall \, E \in I \text{ either $\Lambda_{L,x}$ or $\Lambda_{L,y}$ is
$(\mu/8,E)$-regular} \}  \\ \geq 1- 8 \lvert \Lambda_{L,x} \rvert(C\lvert I \rvert  + 4C'\lvert \Lambda_{L,x} \rvert / \pi ) \euler^{-\mu sL /8} .
\end{multline*}
\end{proposition}
In the proof of Proposition \ref{prop:replace-msa} Hypothesis (ii) is only used to obtain a Wegner estimate. In particular, there is a relation between a Wegner estimate and the a priori bound in the fractional moment method. The following proposition states that the boundedness of averaged fractional powers of the diagonal Green function elements implies a Wegner estimate.
\begin{proposition}[\cite{ElgartTV-11}]\label{prp:Wegner-a-priori}
Let $I \subset \RR$ be an interval, $s \in (0,1)$ and $c > 0$. Assume there is a constant $C \in \RR^+$ such that
\[
 \EE \bigl\{ \lvert G_{\omega , L} (E+\i \epsilon ; x,x) \rvert^s \bigr \} \leq C
\]
for all $L \in \NN$, $x \in \Lambda_L$, $E \in I$ and all $\epsilon \in (0,c]$. Then we have for all $[a,b] \subset I$ with $0< b-a \leq c$ that
\begin{equation*} \label{eq:wegner}
\EE \bigl\{ \Tr \chi_{[a,b]}(H_{\omega , {L}})  \bigr\}
\leq \frac{4C}{\pi}  \lvert b-a \rvert^{s}  \lvert \Lambda_{L} \rvert
 .
\end{equation*}
\end{proposition}
\begin{proof}
 Let $[a,b] \subset I$ with $0<b-a
\leq c$. Since we have for any $\lambda \in \RR$ and $0<\epsilon \leq b-a$
\[
 \arctan \left( \frac{\lambda - a}{\epsilon} \right) - \arctan \left( \frac{\lambda - b}{\epsilon} \right) \geq \frac{\pi}{4} \ \chi_{[a,b]}(\lambda) ,
\]
one obtains an inequality version of Stones formula:
\[
 \langle \delta_x , \chi_{[a,b]} (H_{\omega , {L}}) \delta_x \rangle
\leq \frac{4}{\pi} \int_{[a,b]} \im \left\{ G_{\omega , {L}} (E+ \i \epsilon ; x,x) \right\} \drm E \quad \forall \, \epsilon \in (0, b-a] .
\]
Using triangle inequality, $\lvert \im z\rvert \leq \lvert z\rvert$
for $z \in \CC$, Fubini's theorem, $\lvert G_{\omega , {L}} (E+\i
\epsilon ; x,x) \rvert^{1-s} \leq \dist (\sigma(H_{\omega , L}) ,
E+i \epsilon)^{s-1} \leq \epsilon^{s-1}$ and hypothesis  (ii) we
obtain for all $\epsilon \in (0,b-a]$
\begin{align*}
\EE \bigl\{ \Tr \chi_{[a,b]}(H_{\omega , L}) \bigr\} & \leq \EE \Bigl\{ \sum_{x \in \Lambda_{L}} \frac{4}{\pi} \int_{[a,b]} \im \left\{ G_{\omega , L} (E+\i \epsilon ; x,x) \right\} \drm E  \Bigr\} \\
&  \leq  \frac{\epsilon^{s-1}}{\pi / 4}  \sum_{x \in \Lambda_{L}} \int_{[a,b]} \EE \Bigl\{ \bigl|  G_{\omega , L} (E+\i \epsilon ; x,x)  \bigr|^{s} \Bigr\} \drm E   \\
& \leq 4\pi^{-1}\epsilon^{s-1}  \lvert \Lambda_{L} \rvert \, \lvert b-a \rvert C .
\end{align*}
We minimize the right hand side by choosing $\epsilon = b-a$ and obtain the statement of the proposition.
\end{proof}
Let us note that a Wegner estimate implies the boundedness of an averaged fractional power of the (finite-volume) Green function. At the moment we only know a proof where the bound depends polynomially on the volume of the cube.
\par
From the discussion so far it follows that  Hypothesis (ii) of Proposition \ref{prop:replace-msa} can be replaced by a Wegner estimate. 
Specifically, the following assertion holds true.
\begin{proposition}[\cite{ElgartTV-11}] \label{prop:replace-msa-Wegner}
Let $I \subset \RR$ be a bounded interval and $s \in (0,1)$. Assume the following two statements:
\begin{enumerate}[(i)]
\item There are constants $C,\mu \in \RR^+$ and $L_0 \in \NN_0$ such that 
\[
\EE \bigl\{\lvert G_{\omega, \Lambda_{L,k}} (E;x,y) \rvert^{s} \bigr\} \leq C \euler^{-\mu \lvert x-y \rvert_\infty}
\]
for all $k \in \ZZ^d$, $L \in \NN$, $x,y \in \Lambda_{L,k}$ with $\lvert x-y \rvert_\infty \geq L_0$, and all $E \in I$.
\item There are constants $C_{\rm W}\in \RR^+$, $ \beta \in (0,1]$, and $D \in \NN$
such that 
\[
\PP\bigl\{ \sigma(H_{\omega , \Lambda_{L}} ) \cap [a,b]\not=\emptyset\bigr\} \leq C_{\rm W}{\lvert b-a\rvert}^\beta \, L^D
\]
for all $L \in \NN$ and all $[a,b]\subset I$.
\end{enumerate}
Then we have for all $L \geq \max\{8 \ln (2)/\mu , L_0 , -(8/5\mu)\ln (\lvert I \rvert / 2)\}$ and all $x,y \in \ZZ$ with $\lvert x-y\rvert_\infty \geq 2L+\diam \Theta + 1$ that
\begin{multline*}
 \PP \{\forall \, E \in I \text{ either $\Lambda_{L,x}$ or $\Lambda_{L,y}$ is
$(\mu/8,E)$-regular} \}  \\ \geq 1- 8(2L+1)^d\rvert(C \, \lvert I \rvert  + C_{\rm W}L^D  ) \euler^{-\mu \beta L/8} .
\end{multline*}
\end{proposition}

To conclude exponential localization from the estimates provided in Proposition \ref{prop:replace-msa} or  \ref{prop:replace-msa-Wegner} 
we will use Theorem 2.3 in \cite{DreifusK-89}.
More precisely we need a slight extension of the result, 
which can be proven with the same arguments as the original result. 
What matters for the proof of Theorem~\ref{thm:vDK-2.3} is that there is an $l_0 \in \NN$ 
such that potential values at different lattice sites are independent if their distance is larger or equal $l_0$.

\begin{theorem}[\cite{DreifusK-89}] \label{thm:vDK-2.3}
Let $I\subset \RR$ be an interval and let $p>d$, $ L_0>1$, $\alpha \in (1,2p/d)$ and $m>0$.
Set $L_{k} =L_{k-1}^\alpha$, for $k \in \NN$. Suppose that for any $k \in \NN_0$
\[
 \PP \{\forall \, E \in I \text{ either $\Lambda_{L_k,x}$ or $\Lambda_{L_k,y}$ is $(m,E)$-regular} \}\ge 1- L_k^{-2p}
\]
for any $x,y \in \ZZ^d$ with $\lvert x-y\rvert_\infty \geq 2L_k + \diam \Theta + 1$.
Then $H_\omega$ exhibits exponential localization in $I$ for almost all $\omega \in \Omega$.
\end{theorem}
From Proposition \ref{prop:replace-msa} and Theorem \ref{thm:vDK-2.3} we conclude that the decay of fractional moments of the Green function implies exponential localization.
\begin{theorem}[\cite{ElgartTV-11}] \label{thm:result2}
Let $s \in (0,1)$, $C,\mu, \in \RR^+$, and $I \subset \RR$ be a interval. Assume that
\[
\EE \bigl\{ \lvert G_{\omega , \Lambda_{L,k}} (E + \i \epsilon;x,y) \rvert^{s} \bigr\} \leq C \euler^{-\mu \lvert x-y \rvert_\infty}
\]
for all $k \in \ZZ^d$, $L \in \NN$, $x,y \in \Lambda_{L,k}$, $E \in I$ and all $\epsilon \in (0,1]$.
Then $H_\omega$ exhibits exponential localization in $I$ for almost all $\omega \in \Omega$.
\end{theorem}
Putting together Theorem~\ref{thm:result1} and Theorem~\ref{thm:result2}, 
one can prove without the use of MSA exponential localization in the case of sufficiently large disorder.
\begin{theorem}[\cite{ElgartTV-11}, Exponential localization via fractional moments] \label{thm:result3}
Let Assumption \ref{ass:monotone} be satisfied and $\lambda$ sufficiently large.
Then $H_\omega$ exhibits exponential localization for almost all $\omega \in \Omega$.
\end{theorem}
%
%
%
%
%
%
\section{Localization via multiscale analysis} \label{sec:msa}
The goal of this section is to explain how to obtain
localization for the discrete alloy-type model using multiscale analysis under the assumptions
\begin{itemize} 
 \item $u (0) \not = 0$, $\lvert u(k) \rvert \leq \euler^{-c\lvert k \rvert_\infty}$ for some positive constant $c$, and
 \item the measure $\mu$ has a bounded density $\rho$ with $\supp \rho \subset [-1,1]$.
\end{itemize}
For background
on multiscale analysis, we refer for example to
\cite{Kirsch-08} and \cite{Stollmann-01}. The main
idea of multiscale analysis is to show that Green's function
decay on a small scale implies Green's function decay
on a larger scale. In order to quantify decay of
the Green's function, we introduce

\begin{definition}
 Let $\gamma > 0$. A box $\Lambda_{r,n} \subseteq \mathbb{Z}^d$ is called
 $\gamma$-suitable for $H_{\omega} - E$ if
 \begin{enumerate}
  \item[(i)] $\|(H_{\omega,\Lambda_{r,n}} - E)^{-1}\| \leq \mathrm{e}^{\sqrt{r}}$.
  \item[(ii)] For $x,y\in\Lambda_{r,n}$ with $|x-y|\geq \frac{r}{10}$,
   we have
   \be \nonumber
    |G_{\omega , \Lambda_{r,n}}(E; x,y)|\leq\frac{1}{\#(\Lambda_{r,n})} 
    \mathrm{e}^{-\gamma |x-y|} .
   \ee
 \end{enumerate}
\end{definition}

Only condition (ii) is needed for multiscale analysis when a Wegner
estimate is available. However, our goal is to do without it and then 
(i) takes the role of the Wegner estimate.

The essential step in \cite{Krueger} to conclude exponential
and dynamical localization is

\begin{theorem}\label{thm:mainkmsa}
 Given $\gamma > 0$ and $E \in \mathbb{R}$,
 there is $R = R(\gamma, u, \mu) \geq 1$ such that if
 \be\label{eq:kmsaestimate}
  \mathbb{P}\{\Lambda_{r}\text{ is $\gamma$-suitable for $H_{\omega} - E$}\} 
   \leq \frac{1}{r^{4d}}
 \ee
 holds for $1 \leq r \leq R$ then it holds for all $r > 0$ with $\gamma$ replaced
 by $\gamma / 2$.
\end{theorem}

This theorem implies localization at large coupling:
\begin{itemize}
 \item For large coupling, the Combes--Thomas estimate and a probabilistic
  computation imply the assumptions of Theorem~\ref{thm:mainkmsa}. This can be found
  in Appendix~A of \cite{Krueger}.
 \item The conclusions of this theorem imply localization.
  Exponential localization is proven  in Section~7 of \cite{BourgainK-05}
  and dynamical localization in Sections 16 to 18 of \cite{Krueger}.
\end{itemize}
\begin{theorem}[\cite{Krueger}]
 Let $\lambda > 0$ be large enough. Then, for almost all $\omega \in \Omega$, $H_\omega$ exhibits dynamical (and spectral) localization.
\end{theorem}

We will now discuss how to prove Theorem~\ref{thm:mainkmsa}.
The main difficulty is to conclude from a probabilistic estimate on
a small scale, i.e.\ \eqref{eq:kmsaestimate} for some $r \geq 1$,
the estimate on the resolvent with a better probability, i.e.
\be\label{eq:kmsaestimateR}
 \mathbb{P}\{\|(H_{\omega,R} - E)^{-1}\| > \mathrm{e}^{\sqrt{R}}\} < \frac{1}{R^{4d}}.
\ee
It turns out, that we do not know how to do this. Instead we assume
that \eqref{eq:kmsaestimate} holds for all values of $r$ in a range $[r_1, r_2]$
with $r_2 = (r_1)^3$ and use this to conclude \eqref{eq:kmsaestimateR}.

Just using a single value of $r$ one can conclude that there exists
an event $\mathcal{B}$ with $\mathbb{P}\{\mathcal{B}\} \leq 1 / (2 R^{4d})$,
such that for $\omega \notin\mathcal{B}$, there are
$m_1, \dots, m_{L}$ with $L$ uniformly bounded in $R$
such that for
\be \nonumber
 \Xi = \Lambda_{R} \setminus \bigcup_{\ell=1}^{L} \Lambda_{r, m_{\ell}}
\ee
we have
\be\label{eq:kmsaestiTheta}
 \|(H_{\omega,\Xi} - E)^{-1}\| \leq \mathrm{e}^{3 \sqrt{r}}.
\ee
This follows from the probabilistic estimates in
Section~5 and the estimates on the resolvent in Section~9
of \cite{Krueger}.

Then one uses the assumption on the whole range $[r_1, r_2]$ to conclude
that there exists a choice $\tilde{\omega}$ which agrees with $\omega$
except near the $m_{\ell}$ such that
\be \nonumber
 \|(H_{\tilde{\omega}, R} - E)^{-1} \| \leq \mathrm{e}^{10 \sqrt{r_2}}.
\ee
We illustrate this in the next subsection.
See  Section~12 of \cite{Krueger} or Section~2 in \cite{Bourgain-09}
for the entire implementation.

In Subsection~\ref{sec:appcartan}, we illustrate how to use
Cartan's lemma to conclude from this information that
the estimate \eqref{eq:kmsaestimateR} holds. 
For the entire analysis, see Sections~13 to 15 in \cite{Krueger}.

Then one can conclude the decay of the Green function as in
the usual multiscale analysis finishing the proof
of Theorem~\ref{thm:mainkmsa}.

\subsection{Probabilistic estimates}

In this section, we illustrate a new form of probabilistic estimate
not necessary in the usual versions of multiscale analysis.
We will assume for $r_{\ell} = r \cdot \ell$ with $\ell = 1, \dots, L$
that
\be \nonumber
 \mathbb{P} \{\|(H_{\omega, r_{\ell}} - E)^{-1}\| > A\} \leq \varepsilon
\ee
and that $\supp(u) \subseteq \Lambda_{r-1}$. The main conclusion
will be

\begin{lemma}\label{lem:kmsaprob2}
 There exists an event $\mathcal{B}$ with
 \be \nonumber
  \mathbb{P}\{\mathcal{B}\} \leq \varepsilon^L
 \ee
 such that for $\omega\notin\mathcal{B}$, there exists $\ell$
 such that there exists $\tilde{\omega}$ with
 \be \nonumber
 \omega_n = \tilde{\omega}_n,\quad n \notin \Lambda_{r_{\ell-1}}
 \ee
 we have
 \be \nonumber
  \|(H_{\tilde{\omega}, r_{\ell}} - E)^{-1} \| \leq A .
 \ee
\end{lemma}

\begin{proof}
 Denote by $X_\ell$ the set of all $\omega$ such that 
 for all $\tilde{\omega}$ with
 \[
  \omega_n = \tilde{\omega}_n,\quad n \notin \Lambda_{r_{\ell-1}}
 \]
 we have
 \[
  \|(H_{\tilde{\omega}, r_{\ell}} - E)^{-1}\| > A.
 \]
 By assumption, we have that $\mathbb{P}\{X_{\ell}\} \leq \varepsilon$
 and one can check that $X_{\ell}$ and $X_j$ are independent
 events for $j \neq \ell$.

 Take $\mathcal{B} = \bigcup_{\ell=1}^{L} X_{\ell}$.
 Independence implies $\mathbb{P}\{\mathcal{B}\} \leq \varepsilon^L$.
 The claim follows by construction.
\end{proof}

This lemma implies that given $m \in \Lambda_R$, one can always
find some $\tilde{r} \leq r^2$ and $\tilde{\omega}$ the resolvent
estimate holds for the cube $\Lambda_{\tilde{r}, m}$. 
Using Theorem~9.4. in \cite{Krueger}, one can then
extend \eqref{eq:kmsaestiTheta} to the estimate
\be \nonumber
 \| (H_{\tilde{\omega}, R} - E)^{-1} \| \leq \mathrm{e}^{r}
\ee
for some $\tilde{\omega}$ for which $\tilde{\omega}_n = \omega_n$
whenever $|n-m_{\ell}| \geq r^2$.


\subsection{An application of Cartan's lemma}\label{sec:appcartan}
In this section, we illustrate the application of Cartan's Lemma
with a simple application to random Schr\"odinger operators.
The main goal is to give a simplified account of what happens
in Sections~11 to 14 of \cite{Krueger}. 

Suppose we want to establish a bound on $\|(H_{\omega, L} - E)^{-1}\|$
and that we already know
\begin{enumerate}
 \item For each $\omega$, there exists $n \in \Lambda_L$ such that
  \[
   \|(H_{\omega , \Lambda_L \setminus \Lambda_{r,n}} - E)^{-1}\| \leq A.
  \]
 \item For each $n$ and $\omega$, there exists $\tilde{\omega}$  with
  \[
   \omega_m = \tilde{\omega}_m,\quad m\in \Lambda_L\setminus \{n\}
  \]
  and
  \[
   \|(H_{\tilde{\omega} , L} - E)^{-1}\| \leq A.
  \]
\end{enumerate}
This is a simplification: First, one must allow
for more exceptional sites. Second, such a statement can
allow hold in probabilistic terms in a multi-scale scheme.
\cite{Bourgain-09} was the first paper to propose a scheme to
check (2).

For simplicity, we will also assume that
\begin{enumerate}
 \item[(3)] $\supp(u)$ is contained in $\Lambda_{r}$.
\end{enumerate}
The analysis of $u$ exponential decaying requires a
perturbative analysis, which we avoid here for the sake
of exposition.

\begin{proposition}
 Assume (1), (2), and (3). Then
 \be \nonumber
  \mathbb{P}\{\|(H_{\omega,L} - E)^{-1}\| \geq \mathrm{e}^{s}\}
   \leq \mathrm{e}^{-\frac{\delta s}{\log(A) \cdot r^d}}
 \ee
 for some small constant $\delta > 0$.
\end{proposition}

The proof of this proposition will be split into two parts. First,
we consider the case when the single-site potential $u$ is equal
to $\delta_0$. Second, we discuss what needs to be modified
for general $u$.
We can fix $n$. Then the claim follows by summing over
the possible number of choices of $n$ (less than $(3L)^d$ many).

Assume now that $u = \delta_0$, then we can write $H_{\omega,L} - E$ as
\be \nonumber
 H_{\omega,L} - E = \begin{pmatrix} \omega_n - E & \Gamma \\ \Gamma^{\ast}
  & H_{\omega , \Lambda_L\setminus\{n\}} - E \end{pmatrix}.
\ee
An application of the Schur-complement formula
shows that
\be \nonumber
 \|(H_{\omega,L} - E)^{-1}\| \leq C \cdot A^2
  |\omega_n - E - \Gamma (H_{\omega , \Lambda_{L}\setminus\{n\}} - E)^{-1} \Gamma^{\ast}|^{-1},
\ee
for some constant $C$.
Since the dependence in $\omega_n$ is linear, it is easy to see that
the set of $\omega_n$, where the right hand side is small, is small. Hence,
we are done.

It is noteworthy that this argument did not use (2). Of course,
(2) and not even (1) is necessary since the standard proof
of Wegner's estimate works.

Let us now discuss what happens when $u$ is not $\delta_0$, but is
supported on finitely many points. One could apply the Schur-complement
formula as before, but the object one then obtains has a too non-trivial
dependence on $\omega_n$ to be useful. Define
\be \nonumber
 \tilde\Xi = \Lambda_{r,n} \cap \Lambda_L,\quad \tilde\Theta = \Lambda_L \setminus \tilde\Xi.
\ee
Then $H_{\omega}^{\tilde\Theta} -E$ is independent of $\omega_n$. So an application
of the Schur-complement formula shows that
\be \nonumber
 \|(H_{\omega, L} - E)^{-1}\|\leq C A^2
  \|(H_{\omega , \tilde\Xi} -  E - \Gamma (H_{\omega , \tilde\Theta} - E)^{-1} \Gamma^{\ast} )^{-1} \|.
\ee
Fix $\{\omega_m\}_{m \neq n}$ and
define a function of single-variable $\omega_n$ by
\be \nonumber
 f(\omega_n) = \det(H_{\omega , \tilde\Xi} -  E - \Gamma (H_{\omega , \tilde\Theta} - E)^{-1} \Gamma^{\ast}).
\ee
Define $R = \#\tilde\Xi$. By (1), we obtain that
\be \nonumber
 |f(x)| \leq C A^{R}
\ee
and by (2) that
\be \nonumber
 |f(\tilde{\omega}_n)| \geq \frac{1}{A^R}.
\ee
Hence, we can apply Cartan's Lemma to obtain
\be \nonumber
 |\{\omega_n: \quad |f(\omega_n)|\leq \mathrm{e}^{-s}\}|
  \leq C \exp\left(-\frac{s}{R \log(A)}\right).
\ee
The claim follows by $R \leq (3r)^d$.
%

%

\begin{thebibliography}{HLMW01}

\bibitem[AEN{\etalchar{+}}06]{AizenmanENSS-06}
M.~Aizenman, A.~Elgart, S.~Naboko, J.~H. Schenker, and G.~Stolz, \emph{Moment
  analysis for localization in random {S}chr\"o{}dinger operators}, Invent.
  Math. \textbf{163} (2006), no.~2, 343--413.

\bibitem[Aiz94]{Aizenman-94}
M.~Aizenman, \emph{Localization at weak disorder: some elementary bounds}, Rev.
  Math. Phys. \textbf{6} (1994), no.~5a, 1163--1182.

\bibitem[AM93]{AizenmanM-93}
M.~Aizenman and S.~Molchanov, \emph{Localization at large disorder and at
  extreme energies: An elementary derivation}, Commun. Math. Phys. \textbf{157}
  (1993), no.~2, 245--278.

\bibitem[ASFH01]{AizenmanFSH-01}
M.~Aizenman, J.~H. Schenker, R.~M. Friedrich, and D.~Hundertmark,
  \emph{Finite-volume fractional-moment criteria for {A}nderson localization},
  Commun. Math. Phys. \textbf{224} (2001), no.~1, 219--253.

\bibitem[BGS02]{BourgainGS-02}
J.~Bourgain, M.~Goldstein, and W.~Schlag, \emph{Anderson localization for
  {S}chr\"odinger operators on {$\mathbb Z^2$} with quasi-periodic potential},
  Acta Math. \textbf{188} (2002), no.~1, 41--86.

\bibitem[BK05]{BourgainK-05}
J.~Bourgain and C.~E. Kenig, \emph{On localization in the continuous
  {A}nderson-{B}ernoulli model in higher dimension}, Invent. Math. \textbf{161}
  (2005), no.~2, 389--426.

\bibitem[BLS08]{BakerLS-08}
J.~Baker, M.~Loss, and G.~Stolz, \emph{Minimizing the ground state energy of an
  electron in a randomly deformed lattice}, Commun. Math. Phys. \textbf{283}
  (2008), no.~2, 397--415.

\bibitem[Bou04]{Bourgain-04}
J.~Bourgain, \emph{Recent progress in quasi-periodic lattice {S}chr\"odinger
  operators and {H}amiltonian partial differential equations}, Uspekhi Mat.
  Nauk \textbf{59} (2004), no.~2(356), 37--52.

\bibitem[Bou07]{Bourgain-07}
\bysame, \emph{Anderson localization for quasi-periodic lattice {S}chr\"odinger
  operators on {$\mathbb Z^d$}, {$d$} arbitrary}, Geom. Funct. Anal.
  \textbf{17} (2007), no.~3, 682--706.

\bibitem[Bou09]{Bourgain-09}
\bysame, \emph{An approach to {W}egner's estimate using subharmonicity}, J.
  Stat. Phys \textbf{134} (2009), no.~5-6, 969--978.

\bibitem[CKM87]{CarmonaKM-87}
R.~Carmona, A.~Klein, and F.~Martinelli, \emph{{A}nderson localization for
  {B}ernoulli and other singular potentials}, Commun. Math. Phys. \textbf{108}
  (1987), no.~1, 41--66.

\bibitem[DK89]{DreifusK-89}
H.~von Dreifus and A.~Klein, \emph{A new proof of localization in the
  {A}nderson tight binding model}, Commun. Math. Phys. \textbf{124} (1989),
  no.~2, 285--299.

\bibitem[DK91]{DreifusK-91}
\bysame, \emph{Localization for random {Schr\"odinger} operators with
  correlated potentials}, Commun. Math. Phys. \textbf{140} (1991), no.~1,
  133--147.

\bibitem[dRJLS96]{RioJLS1996}
R.~del Rio, S.~Jitomirskaya, Y.~Last, and B.~Simon, \emph{Operators with
  singular continuous spectrum, {IV}. {H}ausdorff dimensions, rank one
  perturbations, and localization}, J. Anal. Math. \textbf{69} (1996), no.~1,
  153--200.

\bibitem[DS87]{DobrushinS-87}
R.~L. Dobrushin and S.~Shlosman, \emph{Completely analytical interactions:
  Constructive description}, J. Stat. Phys. \textbf{46} (1987), no.~5--6,
  983--1014.

\bibitem[EHa]{ErdoesH-c}
L.~{Erd\H os} and D.~{Hasler}, \emph{Anderson localization at band edges for
  random magnetic fields}, arXiv:1103.3744v1 [math-ph].

\bibitem[EHb]{ErdoesH-b}
\bysame, \emph{Anderson localization for random magnetic {Laplacian} on
  $\mathbb{Z}^{2}$}, arXiv:1101.2139v1 [math-ph].

\bibitem[EHc]{ErdoesH-a}
\bysame, \emph{Wegner estimate and Anderson localization for random magnetic
  fields}, to appear in Commun. Math. Phys, arXiv:1012.5185v1 [math-ph].

\bibitem[ETV10]{ElgartTV-10}
A.~Elgart, M.~Tautenhahn, and I.~Veseli\'c, \emph{Localization via fractional
  moments for models on $\mathbb{Z}$ with single-site potentials of finite
  support}, J. Phys. A: Math. Theor. \textbf{43} (2010), no.~47, 474021.

\bibitem[ETV11]{ElgartTV-11}
\bysame, \emph{Anderson localization for a class of models with a
  sign-indefinite single-site potential via fractional moment method}, Ann.
  Henri Poincar\'e (2011), DOI: 10.1007/s00023-011-0112-5.

\bibitem[FS83]{FroehlichS-83}
J.~Fr\"o{}hlich and T.~Spencer, \emph{Absence of diffusion in the Anderson
  tight binding model for large disorder or low energy}, Commun. Math. Phys.
  \textbf{88} (1983), no.~2, 151--184.

\bibitem[GK04]{GerminetK-04}
F.~Germinet and A.~Klein, \emph{A characterization of the {A}nderson
  metal-insulator transport transition}, Duke Math. J. \textbf{124} (2004),
  no.~2, 309--350.

\bibitem[Gra94]{Graf-94}
G.~M. Graf, \emph{{A}nderson localization and the space-time characteristic of
  continuum states}, J. Stat. Phys. \textbf{75} (1994), no.~1-2, 337--346.

\bibitem[HK02]{HislopK-02}
P.~D. Hislop and F.~Klopp, \emph{The integrated density of states for some
  random operators with nonsign definite potentials}, J. Funct. Anal.
  \textbf{195} (2002), no.~1, 12--47.

\bibitem[HLMW01]{HupferLMW-01a}
T.~Hupfer, H.~Leschke, P.~M{\"u}ller, and S.~Warzel, \emph{The absolute
  continuity of the integrated density of states for magnetic {S}chr\"odinger
  operators with certain unbounded random potentials}, Commun. Math. Phys.
  \textbf{221} (2001), no.~2, 229--254.

\bibitem[Kir89]{Kirsch-89a}
W.~Kirsch, \emph{Random {Schr\"odinger} operators}, {Schr\"odinger} Operators
  (Berlin) (H.~Holden and A.~Jensen, eds.), Lecture Notes in Physics, {\bf
  345}, Springer, 1989.

\bibitem[Kir08]{Kirsch-08}
W.~Kirsch, \emph{An invitation to random {S}chr\"odinger operators}, Random
  {S}chr\"odinger operators, Panor. Synth\`eses, vol.~25, Soc. Math. France,
  2008, With an appendix by Fr{\'e}d{\'e}ric Klopp, pp.~1--119.

\bibitem[KLNS]{KloppLNS}
F.~{Klopp}, M.~{Loss}, S.~{Nakamura}, and G.~{Stolz}, \emph{Localization for
  the random displacement model}, arXiv:1007.2483v2 [math-ph].

\bibitem[Klo93]{Klopp-93}
F.~Klopp, \emph{Localization for semiclassical continuous random
  {Schr\"odinger} operators {II}: The random displacement model}, Helv. Phys.
  Acta \textbf{66} (1993), no.~7-8, 810--841.

\bibitem[Klo95]{Klopp-95a}
F.~Klopp, \emph{Localization for some continuous random {Schr\"o{}dinger}
  operators}, Commun. Math. Phys. \textbf{167} (1995), no.~3, 553--569.

\bibitem[Klo02]{Klopp-02c}
F.~Klopp, \emph{Weak disorder localization and {L}ifshitz tails: continuous
  {H}amiltonians}, Ann. Henri Poincar\'e \textbf{3} (2002), no.~4, 711--737.

\bibitem[KM82]{KirschM-82b}
W.~Kirsch and F.~Martinelli, \emph{On the spectrum of {Schr\"odinger} operators
  with a random potential}, Commun. Math. Phys. \textbf{85} (1982), no.~3,
  329--350.

\bibitem[KM07]{KirschM-07}
W.~Kirsch and B.~Metzger, \emph{The integrated density of states for random
  {Schr\"odinger} operators}, Spectral Theory and Mathematical Physics,
  Proceedings of Symposia in Pure Mathematics, vol.~76, AMS, 2007,
  pp.~649--698.

\bibitem[KN09]{KloppN-09a}
F.~Klopp and S.~Nakamura, \emph{Spectral extrema and {L}ifshitz tails for
  non-monotonous alloy type models}, Commun. Math. Phys. \textbf{287} (2009),
  no.~3, 1133--1143.

\bibitem[KNNN03]{KloppNNN-03}
F.~Klopp, S.~Nakamura, F.~Nakano, and Y.~Nomura, \emph{Anderson localization
  for 2{D} discrete {S}chr\"odinger operators with random magnetic fields},
  Ann. Henri Poincar\'e \textbf{4} (2003), no.~4, 795--811.

\bibitem[Kr{\"u}]{Krueger}
H.~Kr{\"u}ger, \emph{Localization for random operators with non-monotone
potentials with exponentially decaying correlations}, to appear in Ann. Henri Poincar\'e, DOI: 10.1007/s00023-011-0130-3, arXiv:1006.5233v1 [math.SP].

\bibitem[KSS98a]{KirschSS-98b}
W.~Kirsch, P.~Stollmann, , and G.~Stolz, \emph{Anderson localization for random
  {S}chr\"odinger operators with long range interactions}, Commun. Math. Phys.
  \textbf{195} (1998), no.~3, 495--507.

\bibitem[KSS98b]{KirschSS-98a}
W.~Kirsch, P.~Stollmann, and G.~Stolz, \emph{Localization for random
  perturbations of periodic {S}chr\"odinger operators}, Random Oper. Stochastic
  Equations \textbf{6} (1998), no.~3, 241--268.

\bibitem[KV06]{KostrykinV-06}
V.~Kostrykin and I.~Veseli\'c, \emph{On the {Lipschitz} continuity of the
  integrated density of states for sign-indefinite potentials}, Math. Z.
  \textbf{252} (2006), no.~2, 367--392.

\bibitem[Lev96]{Levin-96}
B.~Ya. Levin, \emph{Lectures on entire functions}, Translations of Mathematical
  Monographs, no. 150, American Mathematical Society, 1996.

\bibitem[LPPV08]{LenzPPV-08}
D.~Lenz, N.~Peyerimhoff, O.~Post, and I.~Veseli\'c, \emph{Continuity properties
  of the integrated density of states on manifolds}, Jpn. J. Math. \textbf{3}
  (2008), no.~1, 121--161.

\bibitem[LPPV09]{LenzPPV-09}
\bysame, \emph{Continuity of the integrated density of states on random length
  metric graphs}, Math. Phys. Anal. Geom. \textbf{12} (2009), no.~3, 219--254.

\bibitem[LPV04]{LenzPV-04}
D.~Lenz, N.~Peyerimhoff, and I.~Veseli\'c, \emph{Integrated density of states
  for random metrics on manifolds}, Proc. London Math. Soc. (3) \textbf{88}
  (2004), no.~3, 733--752.

\bibitem[NSV03]{NazarovSV-03}
F.~Nazarov, M.~Sodin, and A.~Volberg, \emph{Local dimension-free estimates for
  volumes of sublevel sets of analytic functions}, Isr. J. Math. \textbf{133}
  (2003), no.~1, 269--283.

\bibitem[PTV11]{PeyerimhoffTV}
N.~Peyerimhoff, M.~Tautenhahn, and I.~Veseli\'c, \emph{Wegner estimate for
  alloy-type models with sign-changing and exponentially decaying single-site
  potentials}, TU Chemnitz Preprint 9, June 2011.

\bibitem[Sto00]{Stolz-00}
G.~Stolz, \emph{Non-monotonic random {S}chr\"odinger operators: the {A}nderson
  model}, J. Math. Anal. Appl. \textbf{248} (2000), no.~1, 173--183.

\bibitem[Sto01]{Stollmann-01}
P.~Stollmann, \emph{Caught by disorder: Bound states in random media}, Progress
  in Mathematical Physics, vol.~20, Birkh\"auser, 2001.

\bibitem[Sto10]{Stolz2010}
G.~Stolz, \emph{An introduction to the mathematics of {A}nderson localization},
  Lecture notes of the Arizona School of Analysis with Applications, 2010, to
  appear in Contemp. Math., arXiv:1104.2317v1 [math-ph].

\bibitem[SW86]{SimonW-86}
B.~Simon and T.~Wolff, \emph{Singular continuous spectrum under rank one
  perturbations and localization for random {H}amiltonians}, Commun. Pur. Appl.
  Math. \textbf{39} (1986), no.~1, 75--90.

\bibitem[TV10]{TautenhahnV-10}
M.~Tautenhahn and I.~Veseli\'c, \emph{Spectral properties of discrete
  alloy-type models}, Proceedings of the XVIth International Congress on
  Mathematical Physics, Prague, 2009, World Scientific, 2010.

\bibitem[Uek94]{Ueki-94}
N.~Ueki, \emph{On spectra of random {S}chr\"odinger operators with magnetic
  fields}, Osaka J. Math. \textbf{31} (1994), no.~1, 177--187.

\bibitem[Uek00]{Ueki-00}
\bysame, \emph{Simple examples of {L}ifschitz tails in {G}aussian random
  magnetic fields}, Ann. Henri Poincar\'e \textbf{1} (2000), no.~3, 473--498.

\bibitem[Uek04]{Ueki-04}
\bysame, \emph{Wegner estimates and localization for {G}aussian random
  potentials}, Publ. Res. Inst. Math. Sci. \textbf{40} (2004), no.~1, 29--90.

\bibitem[Uek08]{Ueki-08}
\bysame, \emph{Wegner estimate and localization for random magnetic fields},
  Osaka J. Math. \textbf{45} (2008), no.~3, 565--608.

\bibitem[Ves02]{Veselic-02a}
I.~Veseli\'c, \emph{{W}egner estimate and the density of states of some
  indefinite alloy type {S}chr\"odinger operators}, Lett. Math. Phys.
  \textbf{59} (2002), no.~3, 199--214.

\bibitem[Ves07]{Veselic-07b}
I.~Veseli\'c, \emph{{\it Existence and regularity properties of the integrated
  density of states of random {Schr\"odinger} Operators}}, Lecture Notes in
  Mathematics, vol. Vol. 1917, Springer-Verlag, 2007.

\bibitem[Ves10a]{Veselic-10b}
I.~Veseli\'c, \emph{Wegner estimate for discrete alloy-type models}, Ann. Henri
  Poincar\'e \textbf{11} (2010), no.~5, 991--1005.

\bibitem[Ves10b]{Veselic-10a}
I.~Veseli\'c, \emph{Wegner estimates for sign-changing single site potentials},
  Math. Phys. Anal. Geom. \textbf{13} (2010), no.~4, 299--313.

\bibitem[Ves11]{Veselic-11}
I.~Veseli\'c, \emph{Lipschitz-continuity of the integrated density of states
  for {G}aussian random potentials}, Lett. Math. Phys. \textbf{97} (2011),
  no.~1.

\bibitem[Weg81]{Wegner-81}
F.~Wegner, \emph{Bounds on the {DOS} in disordered systems}, Z. Phys. B
  \textbf{44} (1981), no.~1-2, 9--15.

\end{thebibliography}
\newcommand{\etalchar}[1]{$^{#1}$}
\providecommand{\bysame}{\leavevmode\hbox to3em{\hrulefill}\thinspace}
\providecommand{\MR}{\relax\ifhmode\unskip\space\fi MR }
\providecommand{\MRhref}[2]{%
  \href{http://www.ams.org/mathscinet-getitem?mr=#1}{#2}
}
\providecommand{\href}[2]{#2}

\end{document}